\documentclass{sig-alternate-05-2015}
\usepackage{mathtools}

\newtheorem{theorem}{Theorem}[section]
\newtheorem{lemma}[theorem]{Lemma}

\begin{document}

\setcopyright{acmcopyright}






%

\title{Analyzing Branch-and-Bound Algorithms for the Multiprocessor Scheduling Problem}

%
%
%
%
%

\numberofauthors{3} 
%
\author{
%
%
\alignauthor 
Thomas Lively\\
       \affaddr{Harvard University}\\
       \affaddr{Cambridge, MA}\\
\alignauthor
William Long\\
\affaddr{Harvard University}\\
       \affaddr{Cambridge, MA}\\
\alignauthor 
Artidoro Pagnoni\\
\affaddr{Harvard University}\\
       \affaddr{Cambridge, MA}\\
}


\maketitle
\begin{abstract}
The Multiprocessor Scheduling Problem (MSP) is an NP-Complete problem with significant applications in computer and operations systems. We provide a survey of the wide array of polynomial-time approximation, heuristic, and meta-heuristic based algorithms that exist for solving MSP. We also implement Fujita's state-of-the-art Branch-and-Bound algorithm \cite{fujita} and evaluate the benefit of using Fujita's binary search bounding method instead of the Fernandez bound \cite{fernandez}. We find that in fact Fujita's method does not offer any improvement over the Fernandez bound on our data set.

\end{abstract}

%
%

%
%

%
%


\section{Introduction}
The Multiprocessor Scheduling Problem (MSP) is the problem of assigning a set of tasks ${j_1, j_2,...j_n}$ to a set of processors ${p_1, p_2,...p_m}$ in such a way that the makespan, or total time required for the completion of the resulting schedule is as small as possible. The tasks may have arbitrary dependency constraints, so they can be modeled as a DAG in which tasks correspond to vertices, and edges encode dependencies between tasks. MSP has been well studied in both theoretical computer science and operations research. Its applications range from industrial project management to tasking cloud-based distributed systems. 

MSP is one problem in a large taxonomy of scheduling problems. Similar problems take into account heterogeneous processors, multiple resource types, communication cost between processors, and the amount of information known the the scheduler. Work on these variants is described in Section 1.3. We chose to focus our work on the basic MSP instead of one of its more esoteric cousins because we are ultimately interested in doing exactly what the problem describes: scheduling multiprocessors.

Before describing Fujita's branch and bound algorithm and our implementation and analysis of it, we provide an introduction to the terminology and notation used to describe MSP and other scheduling problems. We also give a brief survey of the approximate and exact methods and algorithms used to solve MSP.

\subsection{Graham's Notation}
Graham proposed a widely used notation \cite{graham:notation} for succinctly classifying scheduling problems. In Graham's notation a scheduling problem is described in three fields as in $\alpha | \beta | \gamma$. The $\alpha$ field describes the number of processors, $\beta$ describes task configuration options, and $\gamma$ describes the objective function. 

In particular, $\alpha$ is  $Pn$ if we have $n$ identical processors, $Qn$ if we have $n$ uniform processors meaning that each processor has a different compute speed, and $Rn$ if we have $n$ unrelated processors meaning that each processor has a different compute speed for each task. When there is no $n$, the problem is for any number of processors.

$\beta$ is a set that may contain any number of the following options: $r_j$ if tasks have specified release dates, $d_j$ if they have deadlines, $p_j = x$ if each task has weight $x$, $prec$ if tasks have general precedence constraints, and $pmtn$ if tasks can be preempted, meaning they can be stopped and resumed arbitrarily, even moving to other processors.

Finally, $\gamma$ can be any number of different objective functions including the makespan denoted by $C_{max}$, the mean flow-time (completion time minus release date) denoted by $\sum C_i$, or maximum lateness $L_{max} = max(0, C_i - d_i)$.

\subsection{Model}
For our purposes, we are primarily interested in the NP-hard $Pn| prec | C_{max}$ problem. In this precedence-constrained problem, the task graph can be represented as DAG where each vertex $u$ is associated with a task cost $c(u)$ and each edge $(u, v) \in E$ implies that task $v$ can be started only after $u$ is finished. 

Without loss of generality, we can require that the DAGs we schedule contain a single source vertex and a single sink vertex. If there is no unique sink or source  in the DAG, we can simply append a vertex source $u$ with weight $c(u)=0$ as a predecessor to all vertices with zero in-degree and a sink vertex $v$ with $c(v)=0$ as a successor of all vertices with zero out-degree to enforce this requirement.

We adopt the definitions and notation used by Fujita to describe the problem. The only difference is that Fujita considers a generalization of the MSP in which there is allowed to be a communication cost associated with scheduling a successor task on a different processor than its predecessors. This more realistically models the application of scheduling tasks on modern NUMA machines, but we omit communication costs from our model for simplicity.

In our model, we say that a schedule of our task graph $G$ on $p$ processors is a mapping from a vertex $v$ to a tuple $(p, \tau)$ where $p$ is a processor which will process $v$ on the time interval $[\tau, \tau + c(v)]$.

\vspace{2mm}
$\textbf{Definition 1 (Feasible Solution)\cite{fujita}.}$ A Schedule $f$ is said to be feasible, if it satisfies the following two conditions:
\begin{enumerate}
    \item For any $u,v \in V$, if $f(u) = (p, \tau')$ and $f(v) = (p, \tau'')$, then $\tau' + c(u) \leq \tau''$ or $\tau'' + c(v) \leq \tau'$. 
    \item For any $(u,v) \in E$, if $f(u) = (p', \tau')$ and $f(v) = (p'', \tau'')$, then $\tau'' \geq \tau' + c(u)$
\end{enumerate}
\vspace{2mm}
The makespan of $f$ is defined to be the completion time of the exit task $v$ under schedule $f$. The static cost of a path in $G$ is defined as the summation of the execution costs on the path. A path with a maximum static cost is called a critical path in G. Furthermore, we call $t_{cp}$ the static cost of a critical path in G. Lastly, we define 

\vspace{2mm}
$\textbf{Definition 2 (Topological Sort)\cite{fujita}.}$ A topological sort of $G = (V, E)$ is a bijection $\phi$ from $V$ to ${1, 2,...|V|}$ such that for any $u, v \in V$, if $u$ is a predecessor of $v$, then $\phi(u) < \phi(v)$.
\vspace{2mm}

This representation of the precedence constraints will be useful in describing our Branch-and-Bound algorithm. It also helps us define the concept of a partial solution.

\vspace{2mm}
$\textbf{Definition 3 (Partial Solution)\cite{fujita}.}$ Let the graph $G(V,E)$ represent the precedence constraints. A partial solution $x$ is a feasible schedule for a subset of the vertices in $G$. Let $U$ be this subset, the we have that $\phi(u) < \phi(v)$ $\forall u \in U$ and $\forall v \in V$. 
\vspace{2mm}

We note that a solution or a partial solution can be represented as a permutation of the vertices that it schedules. A permutation uniquely represents a schedule, and a partial permutation uniquely represents a partial schedule. To derive a schedule from a partial permutation of the vertices, we iterate through the permutation and assign each task to the first available machine once all its predecessors have finished their execution. Since we only consider those permutations that form feasible partial schedules, we know when we choose how to assign a task that all of its predecessors have already been assigned in the schedule.

\subsection{Known Solutions}
To contextualize our work in the current state of the field, we mention several other scheduling problems similar to MSP and list their best-known runtimes \cite{brucker:textbook}. While the general $P | prec | C_{max}$ problem is NP-hard, some variants are easily solved while others are polynomial but have very high degree. Among the problems known to be solvable in polynomial time are:
\begin{itemize}
    \item $P | p_i = p; tree | C_{max}$ which Hu \cite{hu} solved in $O(n)$.
    \item $P2 | p_i = 1; prec; r_i | \sum C_i$ for which Baptiste and Timkowski \cite{baptiste} found an $O(n^9)$ solution
    \item $R || \sum C_i$ which was solved in $O(mn^3)$ \cite{bruno}
\end{itemize}

On the other hand, the best-known solutions for variants like $Qm | r_i | C_{max}$ are run in pseudo-polynomial time \cite{lawler}and even simplified versions like $P | p_i = 1; prec | C_{max}$ are known to be NP-hard. 

Solutions to this intractable problem have migrated towards approximation schemes. These schemes fall into three categories. The first category encompasses standalone approximation algorithms for the online problem like the guessing scheme of Albers et al \cite{albers:approx} that accomplishes a $(4/3 + \epsilon)$-competitive algorithm building a polynomial number of $O((m/\epsilon)^{O(log(1/\epsilon)/\epsilon)})$ schedules. Integer Programming approaches have also proven to be feasible for graphs with 30-50 jobs \cite{patterson}. The second category are heuristics based on Graham's original List Scheduling Algorithm \cite{graham:list}. However, the accuracy of these approximation strategies is limited. In fact, it has been shown by Ullman that if an approximation scheme for MSP can achieve better than $(4/3 + \epsilon)$, then it could be shown that $P = NP$ \cite{ullman}. The third category consists of meta-heuristic strategies. We expand on the last two strategies here.

\subsection{List Scheduling}
This algorithm is essentially a greedy strategy that maintains a list of ready tasks (ones with cleared dependencies) and greedily assigns tasks from the ready set to available processors as early as possible based on some priority rules. Regardless of the priority rule, List Scheduling is guaranteed to achieve a $2 - 1/m$ approximation. This result can be proved quite simply:

\begin{lemma}
List Scheduling with any Priority Rule achieves a $(2 - 1/m) OPT$ approximation
\end{lemma}
 
\begin{proof}
Given a scheduling of jobs on $m$ processors with makespan $M$ where the sum of all task weights is $S$, we can choose any path and observe that at any point in time, either a task on our path is running on a processor, or no processor is idle. We call $I$ the total idle time and $L$ is the total length of our path. Consequently, we know that: 

    \begin{itemize}
        \item $I \leq (m-1)L$ since processors can be idle only when a task from our path is running.
        \item $L \leq M_{OPT}$ since the optimal makespan is longer than any path in the DAG
        \item $M_{OPT} \geq S/m$ since $S/m$ describes the makespan with zero idle time
        \item $m \times M = I + S$ since the idle time plus sum of all tasks must give us the total "time" given by makespan times number of processors
        \item $m \times M \leq (m-1)M_{OPT} + m M_{OPT}$ implies that $M \leq (2-1/m) M_{OPT}$
    \end{itemize}
\end{proof}

One important priority rule is the Critical Path heuristic which prioritizes tasks on the Critical Path, or longest path from the task to the sink. Other classical priority rules include Most Total Successors (MTS), Latest Finish Time (LFT), and Minimum Slack. Consider, for example, Figure 1.

\begin{figure}
    \centering
    \includegraphics[width=100px]{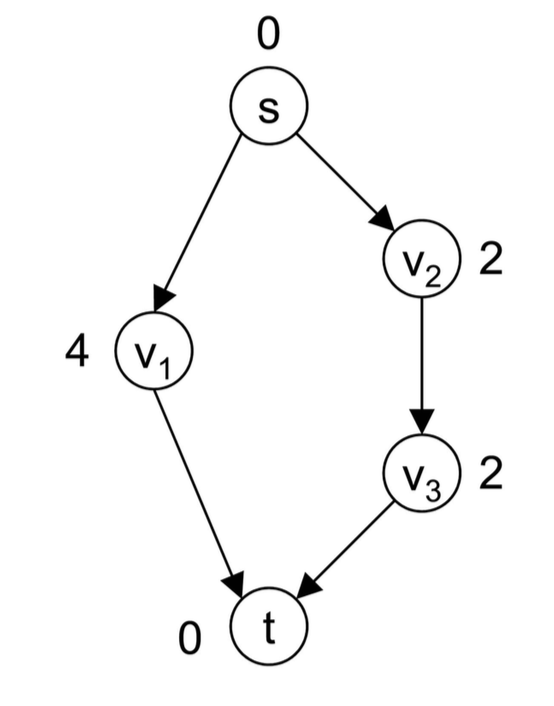}
    \caption{DAG of Tasks and Precedence Constraints}
\end{figure}

When at the source node $s$, List Scheduling would maintain a ready set with tasks $v_1$ and $v_2$. With a Latest Finish Time priority rule, $v_1$ would be first assigned to a processor since it finishes at 4 time steps. With a Critical Path heuristic, either task could be selected since the maximum-length path to the sink vertex is 4 for any path taken. 

Kolisch \cite{kolisch:priority} gives an analysis of four modern priority rules: Resource Scheduling Method (RSM), Improved RSM, Worst Case Slack (WCS), and Average Case Slack (ACS) with better experimental accuracy. In particular, he found that WCS performed best, followed by ACS, IRSM, and LFT. Our List Scheduling implementation utilizes this type of priority rules and attempts to improve upon them by combining them with a branch-and-bound algorithm.

\subsection{Meta-Heuristics}
More recently, research has moved towards using meta-heuristics, a high-level problem-independent algorithmic framework that provides a set of guidelines or strategies to develop heuristic optimization algorithms. For MSP, several strategies have been proposed including utilizing simulated annealing \cite{bouleimen:99}, genetic algorithms \cite{auyeung:genetic}, and even Ant-Colony optimization \cite{selvan}. 

While these meta-heuristics can provide modest improvements in most cases, the largest increases in efficiency are accomplished when heuristics are customized to the MSP problem structure. These meta-heuristics also fail to give a guarantee on the quality of the result, and can converge to local optima. While meta-heuristics can give decent approximations in sub-exponential time, in some situations, obtaining an exact optimal solution is desirable. 

\section{BRANCH AND BOUND METHOD}
The branch-and-bound (BB) method, which is essentially a search algorithm on a tree representing an expansion of all possible assignments, provides an exact solution to MSP. In general, the BB method attempts to reduce the number of expanded sub-trees by pruning the ones that will generate worse solutions than the current best solution. This reduces the number of solutions explored, which would otherwise grow with factorial of the number of nodes. 

Given a graph $G(V,E)$, with an associated partial ordering $\phi$ we can construct the following search tree. The source of the tree is a partial solution only containing the source node of the graph. Each node in the tree corresponds to a partial solution $x$ with respect to a subset $U \subseteq V$, under the form of a permutation of vertices. This means that $x$ provides a scheduling for the nodes in $U$. The leaf nodes are complete feasible solutions. A children of a partial solution $x$ is itself a partial solution that schedules all nodes according to the solution $x$ and also schedules an additional node. Formally, all children of a partial solution $x$ with respect to a subset $U \subseteq V$ are partial solutions with respect to a subset $U \cup \{u\}$ such that $\phi(v) < \phi(u)$ $\forall v \in U$. This means that each vertex that has all its predecessors already scheduled will lead to a new children node and start a new sub-tree. Many nodes will produce schedules with respect to the same subset of vertices. However, they will represent different permutations of the vertices in the subset. The leaves of the tree will contain all permutations of the vertices that lead to feasible schedules. This derives directly from our construction of the graph.

In the BB method, we explore the tree with a depth first search approach. The initial node is the source of the tree, which only contains the trivial schedule for the source of the graph. We expand subsequent nodes according to a priority rule of the same type of those described above. The priority rule that we adopt in our implementation is HLFET (highest level first). Fujita \cite{fujita} also uses the same priority rule in his implementation of the BB algorithm. Both Adams and Canon have studied the performance and robustness of priority rules \cite{Adam, Canon}  in the context of the List Scheduling algorithm described in the previous section. In both numerical experiments the authors have shown that HLFET performs consistently well. Other priority rules and heuristic methods that produce a better estimations of the best node to expand next have been developed, e.g. genetic, and simulated annealing methods. These algorithms give better results compared to the simple priority rules \cite{Auyeung, bouleimen:99}. They are therefore generally used in approximation algorithms for the MSP problem, such as the Grahm's List Scheduling algorithm \cite{graham:list}. These methods also require a significantly longer computation time compared to HLFET. For the BB algorithm, since the heuristic has to be evaluated at every node of the search tree, such computationally expensive methods do not produce any beneficial results. 

In our implementation, the priority rule HLFET assigns a level to every vertex in the graph. The level of a vertex is defined as the sum of the weights of all vertices along the longest path from the vertex to the sink. The search part of the BB algorithm is therefore a depth-first-search algorithm where the priority of nodes in the queue is determined according to HLFET. At each step the BB algorithm expands the node with highest priority first. Intuitively, this will prioritize nodes that have a long list of dependent tasks. A naive search of this type without any bounding component would require to visit all leaf nodes in the search tree. This corresponds to evaluating the schedule quality of all permutations leading to a feasible result, which grows as $O(n!)$ where $n = |V|$.

The core idea of the branch-and-bound algorithm is to prune off all sub-trees that are guaranteed to generate worse solutions than the current best solution. This will significantly reduce  the number of nodes that are expanded in practice. We now need to find a method that produces such a guarantee - the difficulty being that it has to be a guarantee on all solutions that can be reached in a given sub-tree. In the next section we describe two methods to find a lower bound on the makespan of all complete feasible solutions based on a given partial solution. 

It is interesting to note that the BB algorithm generates the solution produced by Graham's list scheduling algorithm \cite{graham:list} with priority queue HLFET as its first solution. The first path expanded in the BB algorithm is composed of the sequence of ready nodes with highest priority at each step, just like in Graham's list scheduling algorithm. The priority rule ensures that the search starts with a good estimate of the optimal solution, and maximizes the number of sub-trees that are pruned.

\section{Fernandez and Fujita Bounds}
We present here the two lower bounding techniques that we implemented. We first describe the Fernandez bound \cite{fernandez}, which is a generalization of Hu's bound \cite{hu} among others. Then we explain the Fujita Bound \cite{fujita}, which generally produces a better lower bound than Fernandez, but is more computationally expensive. Both of these bounds rely on estimating the minimum number of machines required to keep the makespan under a certain total time.

\subsection{Fernandez Bound}
We first need to define $S_x$ the set of complete feasible solutions that can be reached by expanding a given partial solution $x$. All solutions in $S_x$ are represented by permutations in which the initial vertices are exactly the same vertices as in the permutation representing $x$.

Suppose we are given some partial solution $x$, we will now show how to obtain a lower bound on the makespan of all schedules in $S_x$. Fujita \cite{fujita} does not define the quantities correctly, which is very misleading. We are going to follow the logic and definition directly from Fernandez, but stick to the simpler notation employed by Fujita. Let $\theta$ be a subinterval of $\subseteq [0,t_{cp})$ and let $\sigma \in S_x$ be a permutation defining a complete solution in $S_x$.

Suppose that we want to impose a bound on the makespan. Let this bound be $t_{cp}$, the size of the critical path. We define the absolute minimum start time and maximum end time of a task to be respectively the earliest time a task could start executing given its precedence constraints and the latest completion time of a task in order to ensure that its successors can complete within $t_{cp}$. We will refer to these two quantities as $\text{mnEnd}$ and $\text{mxStart}$. Note that these quantities are completely determined from the graph of precedence constraints and do not depend on the number of machines.

We can formally define $\text{mnEnd}$ and $\text{mxStart}$ recursively, which provides an $O(n)$ method for their computation:

\begin{eqnarray}
    \text{mnEnd}(u) &=& c(u) + \max_{v\in V_p(u)} \text{mnEnd}(v)\\
    \text{mxStart}(u) &=& \min \left\{t_{cp}, \min_{v\in V_s(u)} \text{mxStart}(u)\right\}    
\end{eqnarray}
where $V_s(u)$ and $V_p(u)$ are respectively the set of successors and predecessors of $u$.

To determine the previous quantities but given a partial schedule $x$, we fix the start and end times of the tasks in $x$ and calculate $\text{mxStart}$ and $\text{mnEnd}$ with these additional constraints. For vertices that are not in the partial schedule $x$, we note that $\text{mxStart}$ does not depend $x$. On the other hand, $\text{mnEnd}$ depends on $x$ even for nodes that are not in the partial  $x$. Note that the the dependence on the number of machines only comes from the estimation of the execution times of the tasks in the partial schedule $x$.

Consider schedules in $S_x$. We are interested in finding the minimum active time across all machines during a certain interval $\theta$, while bounding the makespan of a the schedules $\sigma\in S_x$ to $t_{cp}$. We define this quantity as $R(\theta)$, and will refer to it as the minimum density function. We will calculate $R(\theta)$ using the previous definitions of $\text{mnEnd}$ and $\text{mxStart}$, we show the detailed derivation at the end of this section. Given this quantity, we could determine the minimum number of machines needed to terminate in time $t_{cp}$ with the following equation:
\begin{equation}
    m_L(t_{cp}) = \max_{\theta \subseteq [0,t_{cp})}\left\lceil\frac{R(\theta)}{|\theta|}\right\rceil
    \label{eq:1}
\end{equation}
If the number of machines that we have available is greater than $m_L(t_{cp})$, the length of the critical path is the best bound that we can give using this approach. Let $m$ be the number of machines that we are given. If $m_L(t_{cp}) < m$, we can find a better upper bound using the approach described by Fernandez \cite{fernandez}. The Fernandez bound on the makespan is $t_{cp} + q$ where $q$ is defined as:
\begin{equation}
    q = \max_{\theta \subseteq [0,t_{cp})}\left\lceil -|\theta| + \frac{R(\theta)}{m}\right\rceil
\end{equation}
Intuitively, we don't have enough machines to complete in $t_{cp}$. During the interval of time that requires most machines, which is the interval with the largest minimum activity, it will take us more time than $\theta$ since we don't have as many machines. We therefore add this extra work $q$ averaged out across all machines to $t_{cp}$.

\subsection{Fujita Bound}
The bound proposed by Fujita relies on equation \ref{eq:1}. The general idea is that we will vary our bound to calculate $\text{mxStart}$ and $\text{mnEnd}$ and find the largest time such that $m < m_L(t_{cp})$. This will certainly be a lower bound on the makespan, since we find the highest time such that the solution is guaranteed to still not be feasible as we don't have enough machines. The Fujita bound relies on calculating $m_L(T)$ multiple times, and is therefore more computationally intensive.

There are two steps in finding this bound. The first step consists in finding the interval within which the bound lies, and then we use binary search to determine the highest time $T$ such that $m < m_L(t_{T})$. Here again, Fujita made an error which makes the logic of the algorithm wrong (the signs of the inequalities are in the wrong direction).

To find an interval, we evaluate $m_L(t_{cp} + \Delta)$ for $\Delta = 1, 2, 4, 8 ...$, until we get $m_L(T) < m$. This gives us the interval $[t_{cp} + 2^{n-1}, t_{cp}+ 2^n)$ within which the bound lies. We then use binary search in this interval and find the highest time $T$ such that $m < m_L(t_{T})$. This requires a total time of $O(\log\Delta_{\text{final}})$.

\subsection{Minimum density function}
Now we just have to show how to determine the minimum density function $R(\theta)$ given a partial schedule $x$ and a time bound $T$. The minimum density function is the minimum active time across all machines during a certain interval $\theta$, while bounding the makespan of the schedules $\sigma\in S_x$ to $T$. 

Let $A$ be the list of all $\text{mnEnd}(u)$ and let $B$ be the list of all $\text{mxStart}(u)$ $\forall u \in V$. We create a sorted list $C$ by merging in linear time the two sorted lists $A$ and $B$. The two lists $A$ and $B$ are constructed recursively, and are sorted by construction.

We now notice that the density function will change only at the time instances corresponding to elements $t_i$ of $C$. This is because the set of tasks that could intersect the interval $\theta$ change only at time instances $t_i\in C$. Furthermore, as shown by Fernandez and Fujita, both $R(\theta)/|\theta|$ and $R(\theta)/m - |\theta|$ decrease monotonically as we increase $\theta$. We will therefore only consider the elements $t_i$ of $C$ as possible limits for the interval $\theta$. 

We then have that the minimum density function is the minimum intersection between the execution time of jobs and the interval $\theta_{ij}$. The only jobs that will be considered are jobs that are necessarily intersecting the interval. We then only take the minimum intersection for each of them. We then define $A^*$ as the set of tasks $u \in V$ such that $t_i < \text{mnEnd}$ and $B^*$ as the set of tasks such that $t_j > \text{mxStart}$. The intersection $A^* \cap B^*$ is the set of tasks that necessarily intersect the interval $\theta_{ij}$. Using the set $A^*\cap B^*$ we can determine the minimum density function: 
\begin{equation}
    R(\theta_{ij}) = \sum_{u \in A^* \cap B^*} \min \{\text{mnEnd} - t_i, c(u), t_j - \text{mxStart}, t_j -t_i\}
\end{equation}
Where $c(u)$ is the weight of task $u$. We see that for each intersecting job, we take the minimum intersection time to be factored in the minimum density function.

This computation takes $O(n)$ in our implementation, which makes the computation of the Fernandez bound $O(n^3)$. In the Fujita bound, we have to repeat this $O(n^3)$ computation to find the correct interval and to search the optimal time bound. Our implementation is publicly available at \cite{github}

\section{Experiments}
To evaluate our implementation, we run it on DAGs generated with the RanGen project generator \cite{Demeulemeester}. Although RanGen produces problem instances for project scheduling problems that contain multiple resource types, we simply set the number of resources to zero to generate DAGs appropriate for our problem. To control the complexity of the generated DAGs we set the order strength parameter in RanGen to 0.1. The order strength is the number of precedence constraints in the generated DAG divided by the largest possible number of precedence constraints. We found that setting order strength of 0.1 produced reasonable-looking DAGs that had plenty of edges but were still solvable on a reasonable number of machines by our implementation in a reasonable amount of time. Although it is unclear that the quality of our implementation run on randomly generated DAGs exactly corresponds to its quality when run on real problems, we believe that being able to control precisely the size and complexity of our test set lets us more thoroughly evaluate and understand the performance of the algorithm.

\begin{figure}[htpb]
\includegraphics[width=8cm]{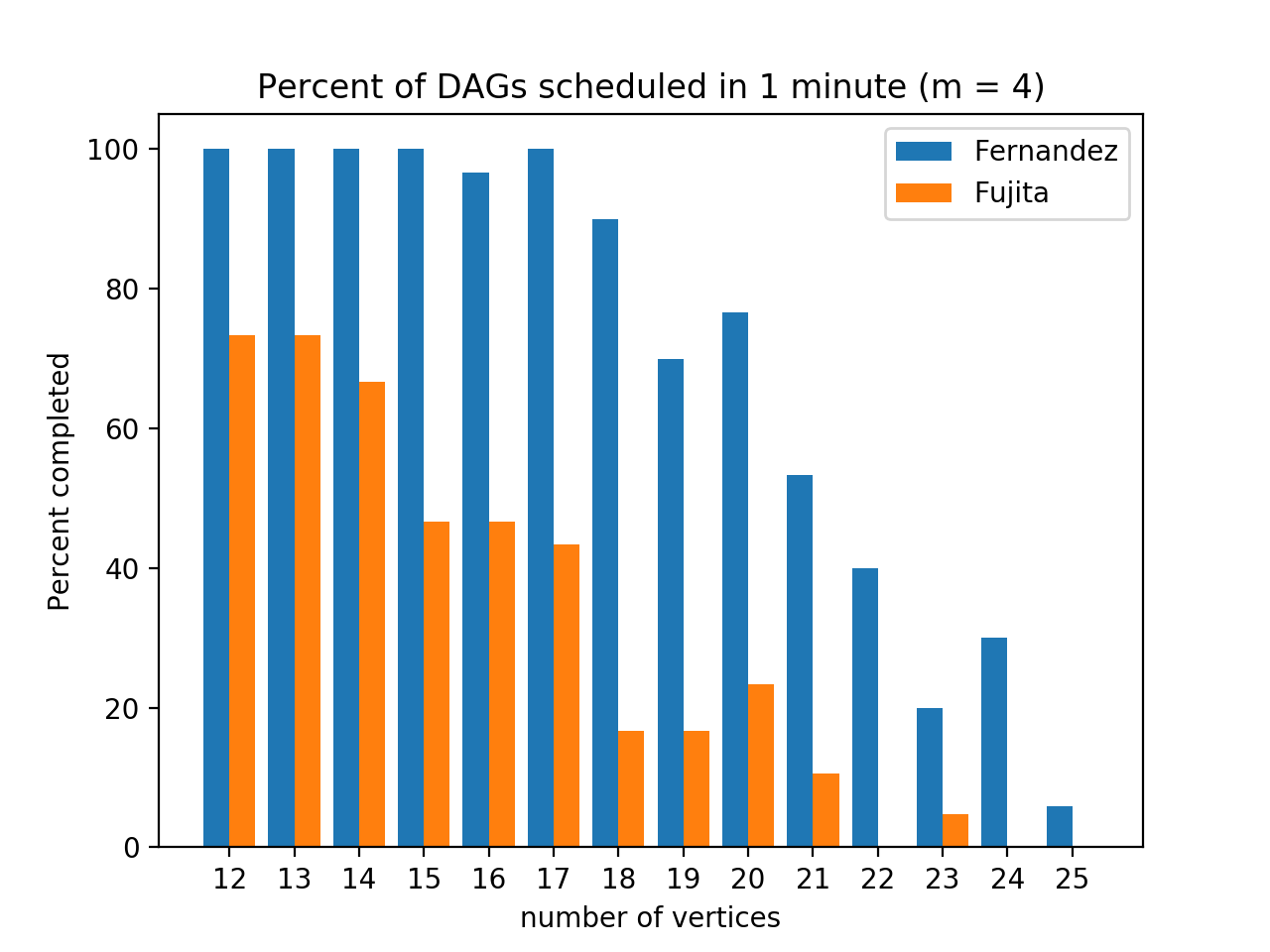}
\centering
\caption{The percent of DAGs of different sizes that were scheduled on four machines in less than one minute using the Fernandez bound and Fujita's binary search bound}
\label{fig:smallcompleted4}
\end{figure}

\begin{figure}[htpb]
\includegraphics[width=8cm]{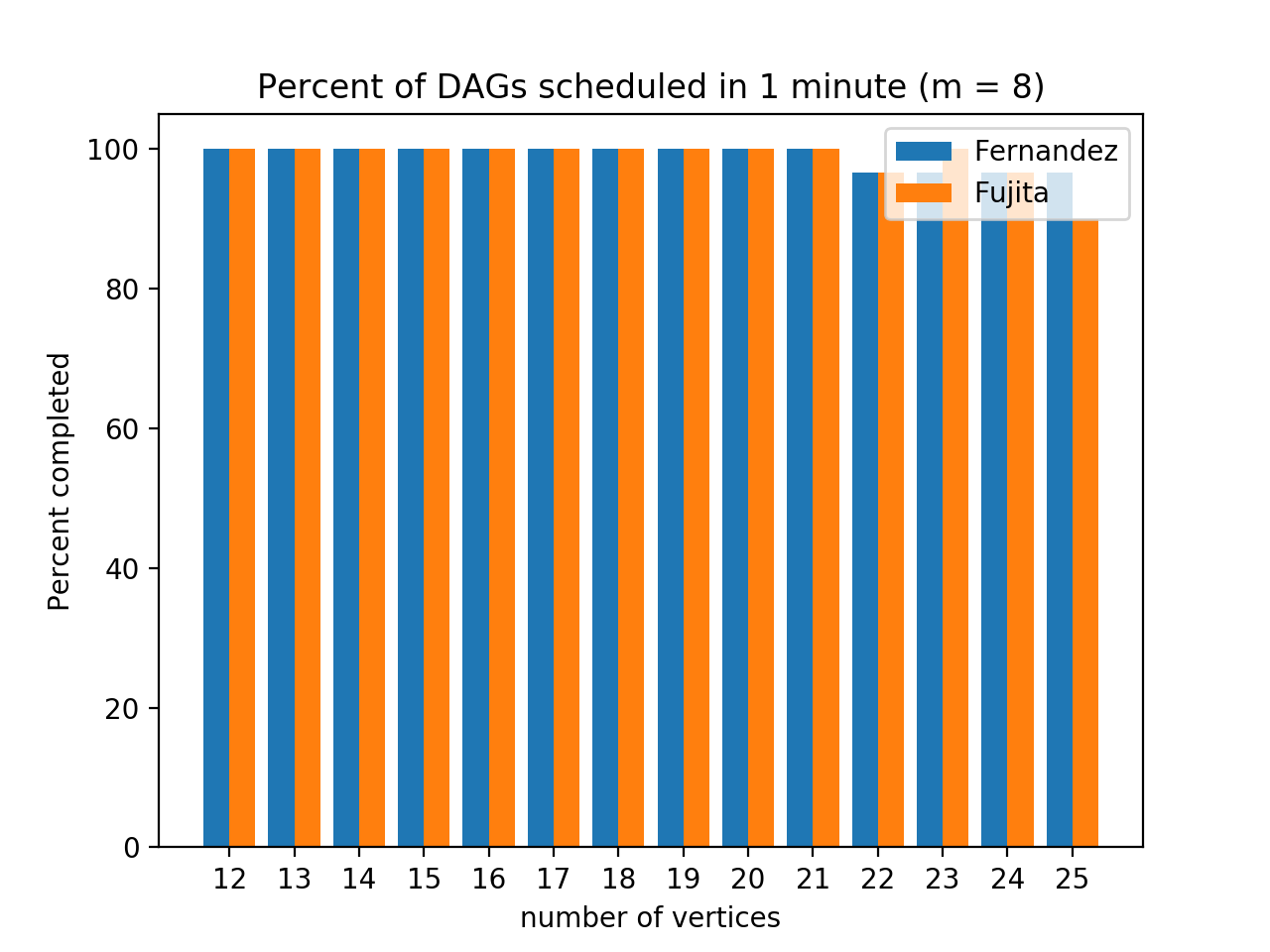}
\centering
\caption{The percent of DAGs of different sizes that were scheduled on eight machines in less than one minute using the Fernandez bound and Fujita's binary search bound}
\label{fig:smallcompleted8}
\end{figure}

\begin{figure}[htpb]
\includegraphics[width=8cm]{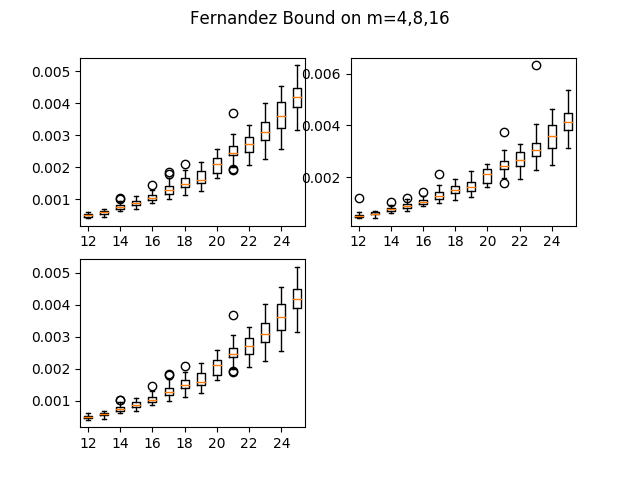}
\centering
\caption{The run times in seconds of successfully scheduled DAGs of different sizes using the Fernandez bound. Top left: m=4. Top right: m=8. Bottom: m=16.}
\label{fig:smallrunFB}
\end{figure}

\begin{figure}[htpb]
\includegraphics[width=8cm]{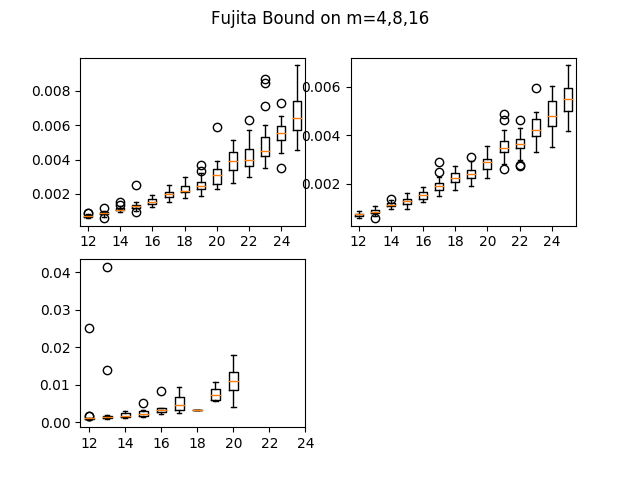}
\centering
\caption{The run times in seconds of successfully scheduled DAGs of different sizes using Fujita's bound. Top left: m=4. Top right: m=8. Bottom: m=16.}
\label{fig:smallrunFujita}
\end{figure}

\begin{figure}[htpb]
\includegraphics[width=8cm]{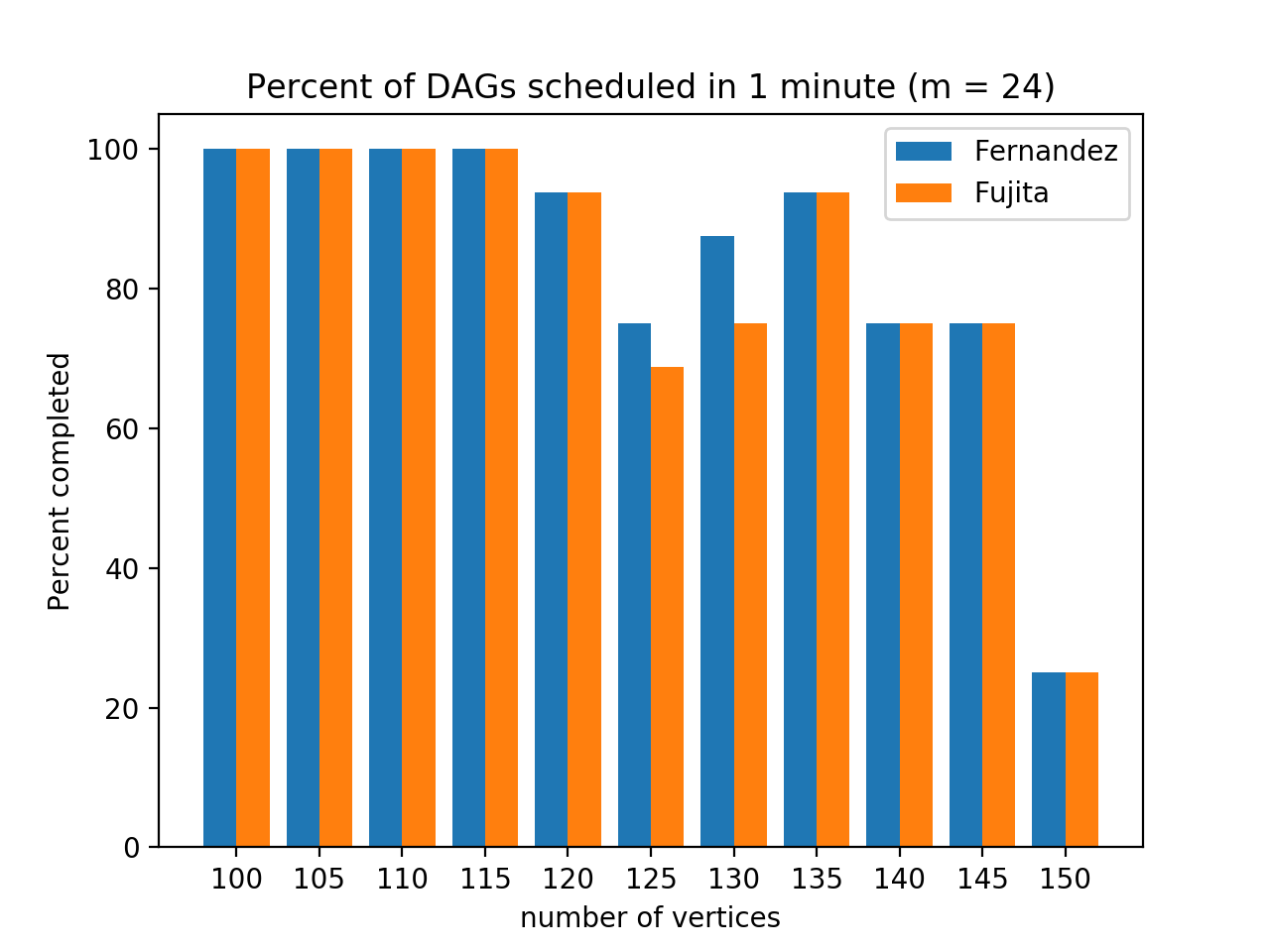}
\centering
\caption{The percent of DAGs of different sizes that were scheduled on 24 machines in less than one minute using the Fernandez bound and Fujita's binary search bound}
\label{fig:largecompleted24}
\end{figure}

\begin{figure}[htpb]
\includegraphics[width=8cm]{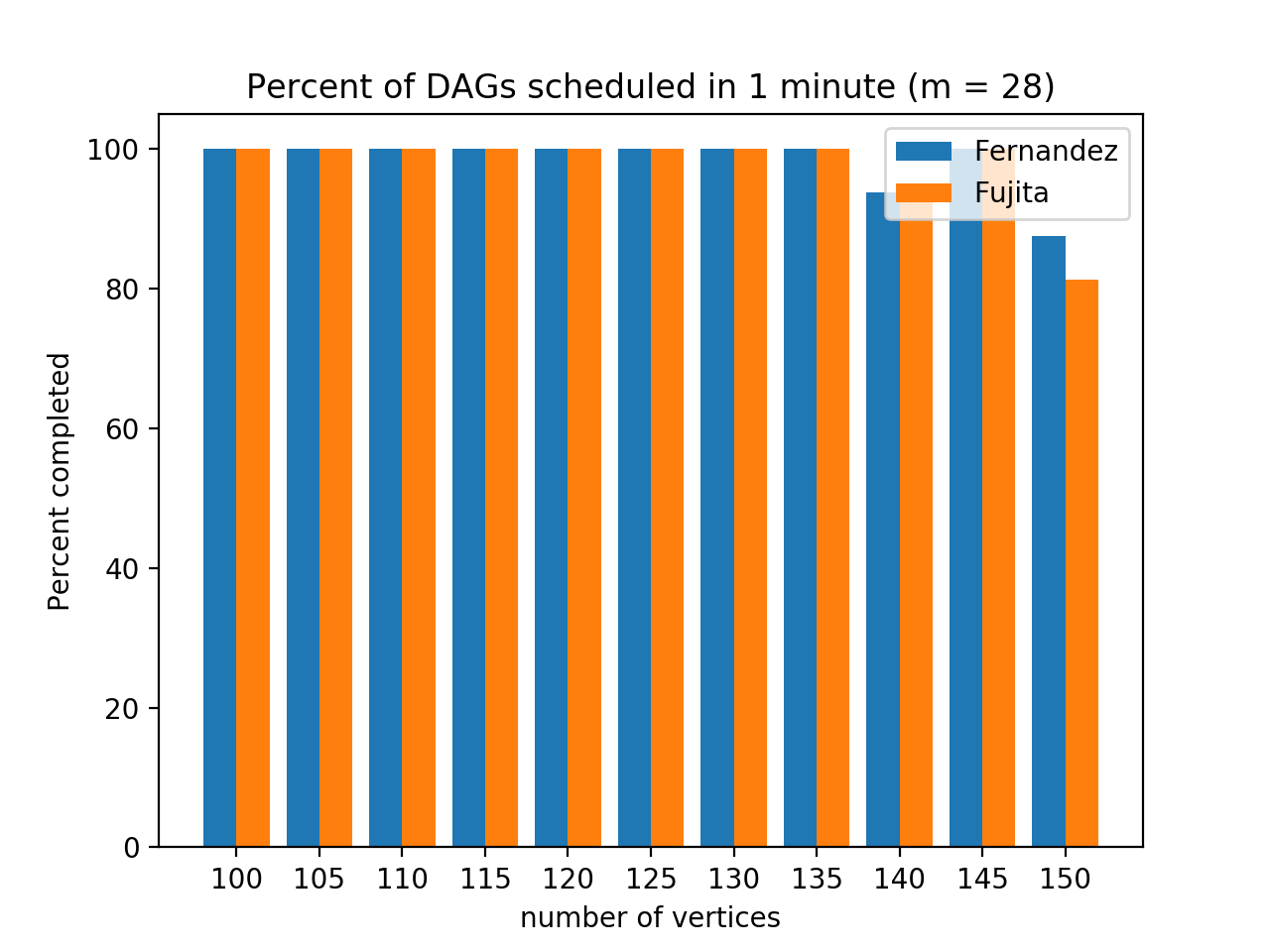}
\centering
\caption{The percent of DAGs of different sizes that were scheduled on 28 machines in less than one minute using the Fernandez bound and Fujita's binary search bound}
\label{fig:largecompleted28}
\end{figure}

\begin{figure}[htpb]
\includegraphics[width=8cm]{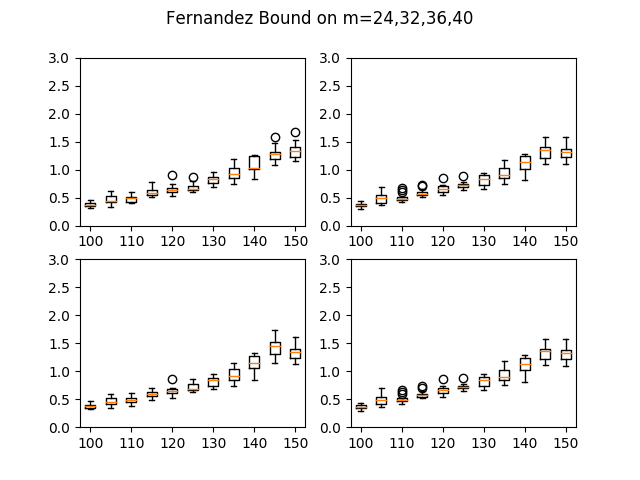}
\centering
\caption{The run times in seconds of successfully scheduled DAGs of different sizes using the Fernandez bound. Top left: m=24. Top right: m=32. Bottom left: m=36. Bottom right: m=40.}
\label{fig:largerunFB}
\end{figure}

\begin{figure}[htpb]
\includegraphics[width=8cm]{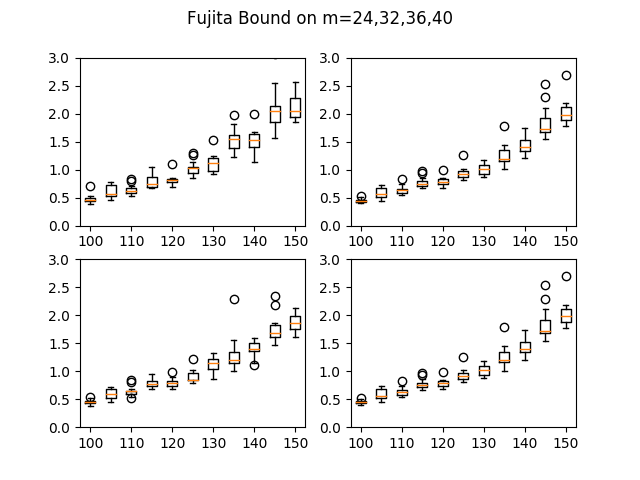}
\centering
\caption{The run times in seconds of successfully scheduled DAGs of different sizes using Fujita's bound. Top left: m=24. Top right: m=32. Bottom left: m=36. Bottom right: m=40.}
\label{fig:largerunFujita}
\end{figure}

Our goals in the experiments are to explore how the runtime of the implementation changes with the inputs to the problem and how Fujita's binary search method for lower bounding the makespan of partial solutions compares to using the Fernandez bound. The first experiment explores the runtime of the algorithm when finding schedules for 4, 8, and 16 machines on DAGs with between 12 and 25 vertices. Figure \ref{fig:smallcompleted4} shows what percent out of thirty DAGs of each size were able to be scheduled on four machines in less than the sixty allotted seconds. Unsurprisingly, the larger the DAG, the harder it is to schedule. However, we were surprised to see that Fujita's binary search bounding method performed worse than just using the Fernandez bound, since Fujita had claimed his method to be an improvement\cite{fujita}.

We were also surprised to find that increasing the number of machines made the scheduling problem easier, though upon reflection this makes sense because having more machines available gives the scheduler more flexibility to make different choices without making the schedule much worse, leading to a better lower bound early on in the execution. Figure \ref{fig:smallcompleted8} shows the percentage of DAGs successfully scheduled in under a minute for eight machines. These are the same DAGs as in Figure \ref{fig:smallcompleted4}, but with eight machines only the largest of the DAGs could not be scheduled. Scheduling for sixteen machines completes in under a minute for all thirty DAGs. For those DAGs that could be scheduled in under a minute, the amount of time each size DAG took to schedule is shown in Figures \ref{fig:smallrunFB} and \ref{fig:smallrunFujita} for the Fernandez and Fujita bounds, respectively. Note that for each DAG, either the DAG is represented in this figure or the DAG took more than sixty seconds to schedule.

The second experiment investigated the execution time of the algorithm for much larger DAGs. Sixteen DAGs each of sizes $100, 105, \dots, 150$ were scheduled on $24, 28, 32, 36, \text{ and } 40$ machines. Any fewer machines and even the 100 vertex DAGs timed out too much to be useful. Overall, the trends seen for large DAGs and large numbers of machines reflect the trends seen with the smaller numbers. Using the Fernandez bound was still more efficient than using Fujita's binary search bounding method, though the gap did seem to close a little bit. It is possible that with even larger graphs using Fujita's method would become beneficial. As with the smaller DAGs, using more machines continued to make the problem easier. Figures \ref{fig:largecompleted24} and \ref{fig:largecomleted28} show the percent of the large DAGs that were successfully scheduled in under a minute on 24 and 28 machines, respectfully. For 32 or more machines, all the DAGs could be scheduled in under a minute. Of those machines that could be scheduled in under a minute, the time it took to schedule each of the large DAG sizes is given in Figures \ref{fig:largerunFB} and \ref{fig:largerunFujita} for the implementation using the Fernandez bound and the implementation with Fujita's bound, respectively.

During development of our implementation we saw that Fujita's binary search bounding method does indeed produce lower bounds at least as good as the Fernandez bound. The only reason the Fernandez bound performs better in our experiments is that Fujita's bound is more computationally complex to calculate. Although the binary search procedure is requires only a number of steps logarithmic in the difference between the lower bound of the current partial schedule and the critical path length of the DAG, each one of those steps requires recomputing the minimum end times, maximum start times, and the minimum work density. Fujita presented a method for calculating the minimum work density in linear time, but our current implementation calculates it in quadratic time. It is therefore possible that reimplementing this calculation to run in linear time would make our implementation using Fujita's bound better than our implementation using the Fernandez bound.

\section{Future work}
One of the most interesting things about the experimental results is that DAGs seem to either be easy or hard to schedule, either taking at most a couple seconds to schedule or taking over sixty seconds. Although there were a few DAGs that took a larger amount of time under sixty seconds to schedule, they were rare. This phenomenon suggests that there might be some way to analyze DAGs and classify them as hard or easy for certain heuristics. If so, the branch and bound algorithm could statically or dynamically choose to use different heuristics for determining the next vertex from the ready set to reduce the number of hard cases.

There are also a number of more immediate ideas we would like to investigate. For example, we would like to quantify how many fewer partial schedules are evaluated when the lower bounding procedure is improved. If we knew how much an improvement in the lower bounding made a difference, we might be able to predict for which DAGs using a more expensive but more exact lower bounding procedure such as Fujita's binary search method would be beneficial.

Finally, we would like to further investigate and compare heuristic algorithms for DAG scheduling. One way we can do this is by halting the branch and bound algorithm after a fixed number of steps and returning the best schedule found so far. Another way is to multiply the lower bound at each step by $(1 + \varepsilon)$ to more aggressively prune the search tree. This would produce an approximation algorithm reaching $(1+\varepsilon)$OPT. It would interesting to compare the computation time of the algorithm using this approximation method compared to other approximation algorithms. Finally, we could investigate improving the branch and bound algorithm performance by implementing multiple list scheduling priority rules, evaluating them, and using them to select new vertices from the ready set in the branch and bound algorithm.

\section{Conclusions}

In this paper we analyze the Multiprocessors Scheduling Problem, and specifically the problem $Pn|prec|c_{max}$ in Graham notation. We describe several approaches used in the literature to solve this $NP$ hard problem. We first explore an approximation algorithm, and then an algorithm that finds the optimal result. In particualar, we derive the $(2-1/m)$OPT bound on the list scheduling algorithm proposed by Graham. We then analyze the Branch-and-Bound method proposed by Fernandez and Fujita, correcting two mistakes in Fujita's exposition of the algorithm.

We have implemented and numerically tested the Branch-and-Bound algorithm, with both the Fernandez bound and the Fujita bound. Experiments were performed on data generated with  RanGen, a tool specifically designed for benchmark tests of scheduling algorithms. With both bounds the algorithm obtains OPT in a few seconds on DAGs of size up to 150 nodes. Our tests demonstrated that Fujita does indeed produce better lower bounds than Fernandez in general. We show however, that this improvement does not justify the increase in computation time.

\bibliographystyle{abbrv}
\bibliography{sigproc}  

\end{document}